\documentclass[a4paper]{article}
\usepackage[top=1 in,bottom=1 in,left=1.0 in,right=1.0 in]{geometry}
\usepackage{amsmath,amsthm}
\usepackage{amssymb}
\usepackage{cite}
\usepackage{graphicx}
\usepackage[colorlinks,linkcolor=blue,citecolor=blue]{hyperref}
\newtheorem{theorem}{Theorem}

\def \t {\widetilde}
\def \h {\widehat}

\title{Bilinearisation-reduction approach to the nonlocal discrete nonlinear Schr\"{o}dinger equations}

\author{Xiao Deng$^{1}$,~ Senyue Lou$^{2}$,~ Da-jun Zhang$^{1}$\footnote{Corresponding author. Email: djzhang@staff.shu.edu.cn, djzhang@shu.edu.cn} \\
$^{1}$Department of Mathematics, Shanghai University, Shanghai 200444, P.R. China\\
$^{2}$Faculty of Science, Ningbo University, Ningbo 315211, P.R. China}

\date{\today}

\begin{document}
\maketitle

\begin{abstract}
  A bilinearisation-reduction approach is described for finding solutions for
  nonlocal integrable systems and is illustrated with nonlocal discrete nonlinear Schr\"{o}dinger equations.
  In this approach we first bilinearise the coupled system before reduction
  and derive its double Casoratian solutions; then we impose reduction on double Casoratians
  so that they coincide with the nonlocal reduction on potentials.
  Double Caosratian solutions of the classical and nonlocal (reverse space,
  reverse time and reverse space-time)  discrete nonlinear Schr\"{o}dinger equations
  are presented.

\begin{description}
\item[MSC:] 35Q51, 35Q55
\item[Keywords:]
nonlocal discrete nonlinear Schr\"{o}dinger equation, bilinear, reduction, double Casoratian solutions
\end{description}
\end{abstract}

\section{Introduction}\label{sec-1}

In recent years the research of nonlocal integrable systems has become a hot topic since  Ablowitz and Musslimani
proposed continuous and discrete nonlocal nonlinear Schr\"{o}dinger (NLS) equations in \cite{AM-PRL-2013,AM-PRE-2014}.
Such systems that are obtained from nonlocal reductions of  integrable  systems
are potentially used to describe wave propagation in nonlinear PT symmetric media \cite{samra-2014}.
Classical solving methods have been successfully applied to nonlocal models (e.g.\cite{AM-Nonl-2016,Zhou-arxiv}).
Recently, a reduction approach was proposed \cite{ChenZh-AML-2017} and solutions of the nonlocal NLS hierarchy
were given in terms of double Wronskian.
In that approach the known double Wronskian solutions of the  coupled  system before reduction
are directly used and by suitable reduction these double Wronskians
yield solutions to the after-reduction nonlocal models.

In general, to solve a nonlocal integrable model, one can first bilinearise a coupled before-reduction system
and write out its double Wronskian (or Casoratian) solutions;
then impose reduction on the double Wronskians (or Casoratians) so that they provide solutions of
the after-reduction equations.
We believe such a bilinearisation-reduction approach  is universal for nonlocal integrable systems.
As a support of this argument in this paper we deal with the nonlocal discrete NLS equations.
Compared with the previous bilinear treatment directly on the nonlocal discrete NLS equation in \cite{zhu-2016},
we will see that our bilinearisation-reduction approach has more advantage in calculation and presenting solutions.

The paper is organized as follows.
In Sec.\ref{sec-2} we introduce the before-reduction coupled system
and after-reduction classical and nonlocal discrete NLS equations.
Sec.\ref{sec-3} provides an illustration of our approach, including bilinearisation of the before-reduction coupled system,
deriving its double Casoratian solutions, and showing the reduction procedure and results.
Sec.\ref{sec-4} is for conclusions.

\section{Nonlocal discrete NLS equations}\label{sec-2}

The integrable discrete NLS equation is a typical and physically important discrete model (see \cite{Abl-PT-book-2004} and the references therein).
It is also known as the Ablowitz-Ladik equation, which is related to the  Ablowitz-Ladik (AL) spectral problem
\begin{equation}\label{AL-sp}
\Theta_{n+1} =M_n  \Theta_n,~~
          M_n=\left(
           \begin{array}{cc}
             \lambda & Q_n \\
             R_n &  1/\lambda \\
           \end{array}
         \right),~~\Theta_n=\left(
           \begin{array}{cc}
             \theta_{1,n} \\
             \theta_{2,n} \\
           \end{array}
         \right),
\end{equation}
where $(Q_n, R_n)=(Q(n,t), R(n,t))$ are potential functions on $\mathbb{Z} \times \mathbb{R}$,
$\lambda$ is a spectral parameter, and $(\theta_{1,n}, \theta_{2,n})$ are wave functions.
The before-reduction coupled system reads
\begin{subequations}\label{nls-bef}
\begin{align}
& i Q_{n,t}=Q_{n+1}+Q_{n-1}-2Q_n -Q_nR_n(Q_{n+1}+Q_{n-1}),\\
& i R_{n,t}=-(R_{n+1}+R_{n-1}-2R_n) + Q_nR_n(R_{n+1}+R_{n-1}).
\end{align}
\end{subequations}
It admits the following reductions (cf.\cite{AM-PRE-2014,AM-SAM-2016})
\begin{equation}
iQ_{n,t}=(Q_{n+1}+Q_{n-1}-2Q_n)-\delta Q_nQ_n^*(Q_{n+1}+Q_{n-1}),~~~~R_n=\delta Q_n^*,
\label{nls-cla}
\end{equation}
\begin{equation}
iQ_{n,t}=(Q_{n+1}+Q_{n-1}-2Q_n) -\delta Q_nQ_{-n}^*(Q_{n+1}+Q_{n-1}),~~~~R_n=\delta Q_{-n}^*,
\label{snls}
\end{equation}
\begin{equation}
iQ_{n,t}=(Q_{n+1}+Q_{n-1}-2Q_n)-\delta Q_nQ_{n}(-t)(Q_{n+1}+Q_{n-1}),~~~~R_n=\delta Q_n(-t),
\label{tnls}
\end{equation}
\begin{equation}
iQ_{n,t}=(Q_{n+1}+Q_{n-1}-2Q_n)-\delta Q_nQ_{-n}(-t)(Q_{n+1}+Q_{n-1}),~~~~R_n=\delta Q_{-n}(-t),
\label{tsnls}
\end{equation}
where $*$ means complex conjugate, $\delta=\pm1$, $Q_{-n}=Q(-n,t)$, $Q_n(-t)=Q(n,-t)$ and $Q_{-n}(-t)=Q(-n,-t)$ indicate
reversed space, time and space-time, respectively.

\section{Bilinearisation-reduction approach}\label{sec-3}

\subsection{Bilinearisation of the before-reduction system}\label{sec-3-1}

In the first step of our approach we bilinearise the before-reduction coupled system \eqref{nls-bef}
rather than those reduced models \eqref{nls-cla}-\eqref{tsnls}.
Through  dependent variables transformation
\begin{eqnarray}\label{trans}
Q_n=\frac{g_n}{f_n},~~R_n=\frac{h_n}{f_n},
\end{eqnarray}
\eqref{nls-bef} is transformed into the following bilinear form
\begin{subequations}\label{nls-bi}
\begin{eqnarray}
  &&iD_{t}g_n\cdot f_n =g_{n+1}f_{n-1}+g_{n-1}f_{n+1}-2g_{n}f_{n}, \label{nls-bi-1}  \\
  &&iD_{t}f_n\cdot h_n =f_{n+1}h_{n-1}+f_{n-1}h_{n+1}-2f_{n}h_{n}, \label{nls-bi-2}  \\
  &&f_n^2-f_{n-1}f_{n+1}=g_nh_n,\label{nls-bi-3}
\end{eqnarray}
\end{subequations}
where $D_t$ is the well known Hirota bilinear operator defined as \cite{Hirota-1974}
\[
D_t^mf\cdot g=(\partial_t-\partial_{t'})^mf(t) g(t')|_{t'=t}.
\]

Let us introduce some notations.
Consider matrix equations
 \begin{subequations}\label{nls-phi-psi}
\begin{align}
&E\Phi_n=A\Phi_n,~~ i\Phi_{n,t}=\frac{1}{2} (E^2-2+E^{-2})\Phi_n,\\
&E^{-1}\Psi_n=A\Psi_n, ~~ i\Psi_{n,t}=-\frac {1}{2} (E^2-2+E^{-2})\Psi_n,
\end{align}
\end{subequations}
where $A\in \mathbb{C}_{(m+p+2)\times(m+p+2)}$ is invertible,
$E$ is a shift operator defined by $E^j f_n=f_{n+j}$, and
\[
\Phi_n=(\phi_{1,n},\phi_{2,n},\cdots,\phi_{m+p+2,n})^T,~~\Psi_n=(\psi_{1,n},\psi_{2,n},\cdots,\psi_{m+p+2,n})^T
\]
are $(m+p+2)$-th order vectors in which $\phi_{j,n}=\phi_j(n,t)$ and $\psi_{j,n}=\psi_j(n,t)$.
In this paper we make use of the following double Casoratian with respect to DOUBLE shifts of $n$,
 \begin{align}
 \mathrm{Cas}^{(m+1,p+1)}(\Phi_n,\Psi_n)& =|\Phi_n,E^2\Phi_n,\cdots,E^{2m}\Phi_n;
 \Psi_n,E^2\Psi_n,\cdots,E^{2p}\Psi_n|=|\h {\Phi_n^{(m)}};\h {\Psi_n^{(p)}}|\nonumber\\
 & =|0,1,\cdots,m;0,1,\cdots,p|=|\h m;\h p|,
 \end{align}
where the shorthand notation ``hat" follows Ref.\cite{Nimmo-NLS-83}.

As for solutions to the before-reduction system \eqref{nls-bef} and its bilinear form  \eqref{nls-bi}, we have the following result:

\begin{theorem}\label{the3}
The bilinear system \eqref{nls-bi} has  double Casoratian solutions
\begin{equation}
 f_n= |\h m; \h p|, ~~ g_n= |\h {m+1}; \h {p-1}|, ~~  h_n= -|\h {m-1}; \h {p+1}|,
 \label{fgh}
 \end{equation}
 where their entries $\Phi_n$ and $\Psi_n$ satisfy \eqref{nls-phi-psi}.
 $A$ and its any similar matrix lead to same solutions to $Q_n$ and $R_n$ through \eqref{trans}.
 \end{theorem}

A brief proof will be sketched in Appendix.

\subsection{Reduction of double Casoratians}\label{sec-3-2}

From now on we consider double Casoratians \eqref{fgh} with $p=m$ and
\begin{eqnarray}\label{phi-psi-A}
\Phi_n=A^ne^{-\frac{i}{2}(A^2-2I_{2(m+1)}+A^{-2})t}C,&\Psi_n=A^{-n}e^{\frac{i}{2}(A^2-2I_{2(m+1)}+A^{-2})t}D,
\end{eqnarray}
which are solutions of \eqref{nls-phi-psi} where $C,D\in \mathbb{C}_{2m+2}$ and $I_s$ is the $s$-th order identity matrix.
The idea of reduction is to impose constraint on $\Phi_n$, $\Psi_n$  so that
$Q_n$ and $R_n$ given by \eqref{trans} can obey the reduction relations in \eqref{nls-cla}-\eqref{tsnls}.

\subsubsection{Discrete NLS equation \eqref{nls-cla}}\label{sec-3-2-1}

We would like to take \eqref{nls-cla} as an example to explain in detail the idea of reduction on double Casoratians.

Consider constraint
\begin{equation}\label{nls-cons}
 \Psi_n = T\Phi^*_n,
\end{equation}
where $T\in \mathbb{C}_{(2m+2)\times (2m+2)}$ is a transform matrix.
It then follows from \eqref{phi-psi-A} and \eqref{nls-cons} that
\begin{align}
  \Psi_n =&A^{-n}e^{\frac{i}{2}(A^2-2I_{2(m+1)}+A^{-2})t}D \nonumber\\
  =&T\Phi_n^*=T A^{*n}e^{\frac{i}{2}(A^{*2}-2I_{2(m+1)}+A^{*-2})t}C^*  \nonumber\\
  =& (TA^*T^{-1})^ne^{\frac{i}{2}((TA^*T^{-1})^2-2I_{2(m+1)}+(TA^*T^{-1})^{-2})t}TC^*,
\end{align}
which requires that $A, T, D, C$ obey the relation
\begin{equation}\label{nls-AT}
A^{-1}=TA^* T^{-1}, ~~D= TC^*
\end{equation}
so that \eqref{nls-cons} holds.
Thus, under \eqref{nls-cons} and \eqref{nls-AT} we can rewrite \eqref{fgh} as
\begin{subequations}\label{nls-fgh}
\begin{eqnarray}
&& f_n=|\h {\Phi_n^{(m)}};\h {\Psi_n^{(m)}}| =
|\h {\Phi_n^{(m)}};\h {T\Phi_n^{*(m)}}|,  \\
&& g_n=|\h {\Phi_n^{(m+1)}};\h {\Psi_n^{(m-1)}}| =
|\h {\Phi_n^{(m+1)}};\widehat{T\Phi_n^{*(m-1)}}|,  \\
&& h_n=-|\h {\Phi_n^{(m-1)}};\h {\Psi_n^{(m+1)}}| =
-|\h {\Phi_n^{(m-1)}};\h {T\Phi_n^{*(m+1)}}|.
\end{eqnarray}
\end{subequations}
Now we introduce another constraint $TT^*=\delta I_{2(m+1)}$ where $\delta=\pm 1$,
by which, making use of determinantal calculation, we find
\begin{eqnarray*}
f_n&=& |\h {\Phi_n^{(m)}};\h {T\Phi_n^{*(m)}}| = |T||\h {\delta T^*\Phi_n^{(m)}};\h {\Phi_n^{*(m)}}|     \\
&=&(-\delta)^{m+1}|T||\h {\Phi_n^{*(m)}};\h {T^*\Phi_n^{(m)}}|  =(-\delta)^{m+1}|T||\h {\Phi_n^{(m)}};\h {T\Phi_n^{*(m)}}|^*  \\
&=&(-\delta)^{m+1}|T|f_n^*,
\end{eqnarray*}
and similarly, $h_n=-(-\delta)^{m}|T|g_n^*$.
Thus we immediately reach
\begin{equation}
\frac{R_n}{Q_n^*}=\frac{h_n/f_n}{g_n^*/f_n^*}=
\frac{h_n}{g_n^*}\cdot\frac{f_n^*}{f_n} =\delta,
\end{equation}
i.e. $R_n=\delta Q_{n}^*$, which is the reduction to get  eq.\eqref{nls-cla} from \eqref{nls-bef}.

In conclusion we have the following.

\begin{theorem}\label{the-nls-cla}
The discrete NLS equation \eqref{nls-cla} has  solutions
\begin{equation}\label{the-nls-Q}
 Q_n= \frac{g_n}{f_n},~~ f_n=|\h {\Phi_n^{(m)}};\h {T\Phi_n^{*(m)}}|,~~
 g_n=|\h {\Phi_n^{(m+1)}};\h {T\Phi_n^{*(m-1)}}|,
\end{equation}
where $\Phi_n$ is given as in \eqref{phi-psi-A}, and $A$ and $T$ obey the relation
\begin{equation}\label{the-nls-AT}
A^{-1}T-TA^*=0,~~ TT^*=\delta I_{2(m+1)}.
\end{equation}
$\Phi_n$  in \eqref{phi-psi-A} has an equivalent form by replacing $A$ with $e^B, ~ B\in \mathbb{C}_{(2m+2)\times (2m+2)}$,
i.e.
\begin{equation}\label{phi-B-nls}
\Phi_n=e^{nB -\frac{i}{2}(e^{2B}-2 I_{2(m+1)}+e^{-2B})t}C.
\end{equation}
Correspondingly, instead of \eqref{the-nls-AT} we need
\begin{equation}\label{the-nls-BT}
BT+TB^*=0,~~ TT^*=\delta  I_{2(m+1)}.
\end{equation}

\end{theorem}

Here we remark that for the continuous focusing NLS equation $iq_t=q_{xx}+|q|^2q$ where $|q|^2=qq^*$, its double Wronskian solution
was given in \cite{Nimmo-NLS-83} in 1983,
while for the discrete NLS equation \eqref{nls-cla},
although there were many discussions on its solutions \cite{Maruno-2006,Maruno-2008,OhtY2014},
surprisingly, it seems there is no explicit double Casoratian form that was presented as its solutions.
Besides, \eqref{the-nls-Q} does not provide a solution with nonzero background for equation \eqref{nls-cla} with $\delta=1$,
i.e. defocusing discrete NLS equation.

\subsubsection{Nonlocal cases}\label{sec-3-2-2}

\begin{theorem}\label{the-snls}
For the nonlocal discrete NLS equation \eqref{snls} with reverse space, its solution is given by
\begin{equation}\label{the-snls-Q}
 Q_n=\frac{g_{n}}{f_{n}},
 ~~
 f_n=|\h {A^{-m}\Phi_n^{(m)}};\h {A^{m}T\Phi_{-n}^{*(-m)}}|,~~
 g_{n}=|\h {A^{-m}\Phi_n^{(m+1)}};\h {A^{m}T\Phi_{-n}^{*(-m+1)}}|,
\end{equation}
where $\Phi_n$ is given as in \eqref{phi-psi-A} or equivalently \eqref{phi-B-nls},
and $A$ and $T$ obey the relation
\begin{equation}\label{the-snls-AT}
AT-TA^*=0, ~~TT^*=-\delta|A^*|^2 I_{2(m+1)},
\end{equation}
or equivalently
\begin{equation}\label{the-snls-BT}
BT-TB^*=0, ~~TT^*=-\delta|e^{B^*}|^2 I_{2(m+1)}.
\end{equation}
Here the notation $\h {A^{m}T\Phi_{-n}^{*(-m)}}$ stands for consecutive columns
$(A^{m}T\Phi_{-n}^{*}, A^{m-2}T\Phi_{-n}^{*}, \cdots, A^{-m}T\Phi_{-n}^{*})$.

For the nonlocal discrete NLS equation \eqref{tnls} with reverse time, its solution is given by
\begin{equation}\label{the-tnls-Q}
 Q_n=  \frac{g_n}{f_n},~~ f_n=|\h {\Phi_n^{(m)}};\h {T\Phi_n^{(m)}(-t)}|,~~
 g_n=|\h {\Phi_n^{(m+1)}};\h {T\Phi_n^{(m-1)}(-t)}|,
\end{equation}
where $\Phi_n$ is given as in \eqref{phi-psi-A} or equivalently \eqref{phi-B-nls},
and $A$ and $T$ obey the relation
\begin{equation}\label{the-tnls-AT}
A^{-1}T-TA=0, ~~ T^2=\delta I_{2(m+1)},
\end{equation}
or equivalently
\begin{equation}\label{the-tnls-BT}
BT+TB=0, ~~T^2= \delta I_{2(m+1)}.
\end{equation}

For the nonlocal discrete NLS equation \eqref{tsnls} with reverse space and time, its solution is given by
\begin{equation}\label{the-tsnls-Q}
  Q_n=\frac{g_{n}}{f_{n}},~~
 f_n=|\h {A^{-m}\Phi_n^{(m)}};\h {A^m T\Phi_{-n}^{(-m)}(-t)}|,
 ~~ g_n= |\h {A^{-m}\Phi_n^{(m+1)}};\h {A^m T\Phi_{-n}^{(-m+1)}(-t)}|,
\end{equation}
where $\Phi_n$ is given as in \eqref{phi-psi-A} or equivalently \eqref{phi-B-nls},
and $A$ and $T$ obey the relation
\begin{equation}\label{the-tsnls-AT}
AT-TA=0,~~T^2=-\delta|A|^2 I_{2(m+1)},
\end{equation}
or equivalently
\begin{equation}\label{the-tsnls-BT}
BT-TB=0, ~~T^2=-\delta|e^B|^2 I_{2(m+1)}.
\end{equation}

\end{theorem}

\begin{proof}

The proof is similar to the classical case.
For the nonlocal discrete NLS equation \eqref{snls} with reverse space,
since $C$ and $D$ in \eqref{phi-psi-A} are arbitrary,
we first replace $C\to A^{-m} C,~ D\to A^{m}D$ and start from
\begin{equation}\label{snls-fgh}
f_n=|\h {A^{-m}\Phi_n^{(m)}};\h {A^{m}\Psi_{n}^{(m)}}|,  ~~
g_n=|\h {A^{-m}\Phi_n^{(m+1)}};\h {A^{m}\Psi_{n}^{(m-1)}}|,  ~~
h_n= -|\h {A^{-m}\Phi_n^{(m-1)}};\h {A^{m}\Psi_{n}^{(m+1)}}|,
\end{equation}
which are still solution to \eqref{nls-bi}. Reduction is implemented by taking
\begin{equation}\label{snls-cons}
 \Psi_n = T\Phi^*_{-n},~~ D=TC^*,
\end{equation}
and requiring  \eqref{the-snls-AT} holds.
Then  we have
\begin{equation}
f_n = |A^{-m}\Phi_n,A^{-m+2}\Phi_n,\cdots,A^{m}\Phi_n;A^{m}T\Phi^*_{-n},A^{m-2}T\Phi^*_{-n},\cdots,A^{-m}T\Phi_{-n}^*|.
\label{fgh-am}
\end{equation}
Making use of \eqref{the-snls-AT} we find
\begin{align*}
f_n =&|A^{-m}\Phi_n,A^{-m+2}\Phi_n,\cdots,A^{m}\Phi_n;T {A^{*}}^{m}\Phi^*_{-n},T {A^{*}}^{m-2}\Phi^*_{-n},\cdots,T {A^*}^{-m}\Phi_{-n}^*|\\
=&|T|(-\delta|A^*|^{-2})^{m+1}| T^*A^{-m}\Phi_n, T^*A^{-m+2}\Phi_n,\cdots,T^*A^{m}\Phi_n;
                                 {A^{*}}^{m}\Phi^*_{-n},{A^{*}}^{m-2}\Phi^*_{-n},\cdots, {A^*}^{-m}\Phi_{-n}^*|   \\
=&(\delta|A^*|^{-2})^{m+1}|T||A^m\Phi_{-n},A^{m-2}\Phi_{-n},\cdots,A^{-m}\Phi_{-n};
                                 T {A^*}^{-m}\Phi^*_{n},T {A^*}^{-m+2}\Phi^*_{n},\cdots,T {A^{*}}^{m}\Phi_{n}^*|^* \\
=&(\delta|A^*|^{-2})^{m+1}|T||A^m\Phi_{-n},A^{m-2}\Phi_{-n},\cdots,A^{-m}\Phi_{-n};
                                 {A}^{-m}T\Phi^*_{n},{A}^{-m+2}T\Phi^*_{n},\cdots,{A}^{m}T\Phi_{n}^*|^* \\
=&(\delta|A^*|^{-2})^{m+1}|T||A^{-m}\Phi_{-n},A^{-m+2}\Phi_{-n},\cdots,A^m\Phi_{-n};
                                 {A}^{m}T\Phi_{n}^*,{A}^{m-2}T\Phi^*_{n},\cdots,{A}^{-m}T\Phi^*_{n}|^* \\
=&(\delta|A^*|^{-2})^{m+1}|T|f_{-n}^*.
\end{align*}
Similarly, we have $h_n=\delta^{m}|A^*|^{-2(m+1)}|T|g_{-n}^*$.
By these relations we immediately reach
${R_n}/{Q_{-n}^*}=\delta$
and then the solution expression \eqref{the-snls-Q}.

For the nonlocal discrete NLS equation \eqref{tnls} with reverse time, we start from
usual forms $f_n=|\h {\Phi_n^{(m)}};\h {\Psi_n^{(m)}}|,~ g_n=|\h {\Phi_n^{(m+1)}};\h {\Psi_n^{(m-1)}}|,~
h_n=-|\h {\Phi_n^{(m-1)}};\h {\Psi_n^{(m+1)}}|$ and reduction is implemented by taking
\begin{equation}\label{tnls-cons}
 \Psi_n = T\Phi_n(-t),~~D=TC
\end{equation}
together with \eqref{the-tnls-AT}.
In this case we have relation
\begin{equation}
f_n=(-\delta)^{m+1}|T|f_n(-t),~~
h_n=-(-\delta)^{m}|T|g_n(-t).
\end{equation}

For the nonlocal discrete NLS equation \eqref{tsnls} with reverse space and time, we need to start from \eqref{fgh-am}
and reduction is implemented by taking
\begin{equation}\label{tsnls-cons}
 \Psi_n = T\Phi_{-n}(-t),~~D=TC
\end{equation}
together with \eqref{the-tsnls-AT}.
Relations between Casorations are
\begin{equation}
f_n =(\delta|A|^{-2})^{m+1}|T|f_{-n}(-t), ~~
h_n = \delta^{m}|A|^{-2(m+1)}|T|g_{-n}(-t).
\end{equation}

\end{proof}

\subsection{Solutions to $B$ and $\Phi_n$}\label{sec-3-3}

In Theorem \ref{the-nls-cla} and \ref{the-snls} we have given constraints on $B$ and $T$.
In the following we present solutions of these constraint equations together with explicit expressions of $\Phi_n$.
For convenience, let us first introduce some notations.
Suppose $B$ and $T_i(i=1,2,3)$ are $2\times 2$ block matrices
\begin{equation*}
B_1= \left(
      \begin{array}{cc}
        K & \mathbf{0} \\
        \mathbf{ 0}& H
      \end{array}
    \right),~~
    T_1= \left(
      \begin{array}{cc}
         \mathbf{0} & I_{m+1} \\
         -I_{m+1} &  \mathbf{0}
      \end{array}
    \right),~~  T_2= \left(
      \begin{array}{cc}
         \mathbf{0} & I_{m+1} \\
         I_{m+1} &  \mathbf{0}
      \end{array}
    \right),~~    T_3= \left(
      \begin{array}{cc}
          I_{m+1}&\mathbf{0} \\
         \mathbf{0}&-I_{m+1}
      \end{array}
    \right),
\end{equation*}
in which each block is a constant $(m+1)\times (m+1)$ matrix; $J(k)$ is a $(m+1)\times (m+1)$ Jordan matrix w.r.t. $k\in \mathbb{C}$,
\begin{equation}
J(k)=\left(
      \begin{array}{cccc}
        k & 0 & \cdots & 0  \\
        1& k& \cdots& 0 \\
        \vdots &\ddots&\ddots& \vdots\\
        0 & \cdots &1&k
      \end{array}\right).
\end{equation}
Then, for the constraint equations  \eqref{the-nls-BT}, \eqref{the-snls-BT}, \eqref{the-tnls-BT} and \eqref{the-tsnls-BT},
their solutions are respectively given by
\begin{align}
& B=B_1~\mathrm{with} ~ H=-K^*,~~ T=T_1 (\mathrm{for} ~ \delta=-1)~\mathrm{or}~T_2 (\mathrm{for} ~ \delta=1)\label{BT-nls}\\
& B=B_1~\mathrm{with} ~ H=K^*,~~ T=|e^{B}| T_2 (\mathrm{for} ~ \delta=-1)~\mathrm{or}~|e^{B}| T_1 (\mathrm{for} ~ \delta=1),\label{BT-snls}\\
& B=B_1~\mathrm{with} ~ H=-K,~~ T=\sqrt{\delta} T_1,\label{BT-tnls}\\
& B=B_1~\mathrm{with} ~ H=-K^*(\mathrm{or}~K^*~\mathrm{or} -K),~~ T=\sqrt{-\delta}|e^{B}| T_3. \label{BT-tsnls}
\end{align}

Due to Theorem \ref{the3} we only need to consider the canonical forms of $B$.
That is, $K$ can either be
\begin{equation}
K=\mathrm{Diag}(k_1, k_2, \cdots, k_{m+1}),~~ k_i\neq 0 \in \mathbb{C},
\label{K}
\end{equation}
or $K=J(k),~ k\neq 0 \in \mathbb{C}$.
When $K$ is diagonal, $\Phi_n$ is composed by
\begin{equation}
\phi_{j,n}=
\left\{\begin{array}{ll}
       e^{k_j n -\frac{i}{2}(e^{2k_j}-2+e^{-2k_j})t}, & j=0, 1, \cdots m,\\
       e^{h_s n -\frac{i}{2}(e^{2h_s}-2+e^{-2h_s})t}, & j=m+1+s,~ s=0, 1, \cdots m,
       \end{array}\right.
\end{equation}
where $h_s=-k_s^*$ for \eqref{nls-cla},  $h_s=k_s^*$ for \eqref{snls},
$h_s=-k_s$ for \eqref{tnls} and   $h_s=-k_s^*(\mathrm{or}~ k_s^*~\mathrm{or}\, -k_s)$ for \eqref{tsnls}.
When $K$ takes the Jordan block $J(k)$,  $\Phi_n$ is composed by
\begin{equation}
\phi_{j,n}=
\left\{\begin{array}{ll}
       \frac{\partial^j_k}{j!} e^{k n -\frac{i}{2}(e^{2k}-2+e^{-2k})t}, & j=0, 1, \cdots m,\\
       \frac{\partial^s_h}{s!} e^{h n -\frac{i}{2}(e^{2h}-2+e^{-2h})t}, & j=m+1+s,~ s=0, 1, \cdots m,
       \end{array}\right.
\end{equation}
where $h_s=-k_s^*$ for \eqref{nls-cla},  $h_s=k_s^*$ for \eqref{snls},
$h_s=-k_s$ for \eqref{tnls} and   $h_s=-k_s^*(\mathrm{or}~ k_s^*~\mathrm{or} -k_s)$ for \eqref{tsnls}.

\subsection{Examples: One-soliton solutions}

As examples we list out  one-soliton solutions we obtained for equations \eqref{nls-cla}-\eqref{tsnls}.

For the classical discrete NLS equation \eqref{nls-cla}, its one-soliton is
\begin{eqnarray}
Q_n=\frac{e^{2k_1}-e^{-2k_1^*}}{e^{-2k_1n+2i\xi_1t}-\delta e^{2k_1^*n+2i\xi^*_1t}};
\end{eqnarray}
for the nonlocal (reverse space) discrete NLS equation \eqref{snls}, we have
\begin{eqnarray}
Q_n=\frac{e^{2k_1}-e^{-2k_1^*}}{e^{k_1+k_1^*}(e^{-2k_1n+2i\xi_1t}+\delta e^{-2k_1^*n+2i\xi^*_1t})};
\end{eqnarray}
for the nonlocal (reverse time) discrete NLS equation \eqref{tnls}, we have
\begin{eqnarray}
Q_n=\frac1{\sqrt{\delta}}\frac{(e^{2k}-e^{-2k})e^{-2i\xi_1t}}{e^{2k_1n}+e^{-2k_1n}};
\end{eqnarray}
for the nonlocal (reverse space-time) discrete NLS equation \eqref{tsnls}, its one-soliton solution can be either of the following,
\begin{subequations}
\begin{eqnarray}
Q_n&=&\frac1{\sqrt{-\delta}}\frac{e^{2k_1}-e^{2k_1^*}}{e^{k_1+k_1^*}(e^{-2k_1^*n+2i\xi_1^*t}+e^{-2k_1n+2i\xi_1t})},\\
Q_n&=&\frac1{\sqrt{-\delta}}\frac{e^{2k_1}-e^{-2k_1}}{(e^{2k_1n}+e^{-2k_1n})e^{2i\xi_1t}},\\
Q_n&=&\frac1{\sqrt{-\delta}}\frac{e^{2k_1}-e^{-2k_1^*}}{e^{k_1-k_1^*}(e^{2k_1^*n+2i\xi_1^*t}+e^{-2k_1n+2i\xi_1t})}.
\end{eqnarray}
\end{subequations}
In the above, $\xi_1=e^{2k_1}-2+e^{-2k_1}$ and $k_1\in \mathbb{C}$.

\section{Conclusion}\label{sec-4}

In this paper we have shown a reduction approach to construct solutions of the classical and nonlocal discrete NLS equations.
In this approach, the first step is to bilinearise the before-reduction system \eqref{nls-bef} and derive its
double Casoratian solutions.
In the second step, by imposing suitable reductions on the Casoratian entries,
one can get relations between $f_n(t)$ and $f^*_{n}(t)$, $f^*_{-n}(t)$, $f_{n}(-t)$ and $f_{-n}(-t)$,
and between $h_n(t)$ and $g^*_{n}(t)$, $g^*_{-n}(t)$, $g_{n}(-t)$ and $g_{-n}(-t)$.
These relations lead to  classical and nonlocal reductions between
$R_n$ and $Q_n$ via the transformation \eqref{trans}.

With regard to the method described in the paper,  bilinearisation is not the only way to solve the before-reduction system.
Obviously, Inverse Scattering Transform and Darboux transformation are the alternative.
We believe the reduction approach is universal for getting solutions of those nonlocal integrable systems.
It is also simpler than the treatment directly working on the after-reduction systems (cf. \cite{zhu-2016,XuTao-2016}).
Dynamics of the obtained variety of solutions and solutions of  these classical and nonlocal discrete NLS  hierarchies
would be investigated in the future.
Besides, it would be interesting to consider possible transformations between classical and nonlocal discrete NLS equations
(cf. for continuous case \cite{YY-arxiv-2016}).

\subsection*{Acknowledgments}
This work was supported by  the NSF of China [grant numbers 11371241, 11435005, 11631007].

\begin{appendix}
\section{Proof of Theorem \ref{the3}}

To prove the theorem, the following identity will be used \cite{FreN1983},
\begin{equation}
|M,a,b||M,c,d|-|M,a,c||M,b,d|+|M,a,d||M,b,c|=0,
\label{id-1}
\end{equation}
where $M$ is a $N\times(N-2)$ matrix and $a,~b,~c,~d$ are $N$-th order column vectors.
Besides, some identities for double Wronskians listed in \cite{Liu1990,YSCC2008} will be also used.

A direct calculation yeilds
\begin{align*}
&f_{n+1}=|A||\h m; \t{p+1}|=|A|^{-1}|\t{m+1};\h p|,~~
 f_{n-1}=|A|^{-1}|\h m;-1, \h{p-1}|,\\
&g_{n+1}=|A||\h{m+1};\t p|,~~
 g_{n-1}=|A||-1,\h{m};\h {p-1}|,\\
&2if_{n,t}=(|\h{m-1},m+1;\h p|-|\h m;-1,\t p|)+(|-1,\t{m};\h p|-|\h m;\h {p-1},p+1|)+2(p-m)f_n,\\
&2ig_{n,t}=(|\h{m},m+2;\h {p-1}|-|\h {m+1};-1,\t {p-1}|) \\
& ~~~~~~~~~~~ +(|-1,\t{m+1};\h {p-1}|-|\h {m+1};\h {p-2},p|)+2(p-m-2)g_n.
\end{align*}
where $|\t {m+1};\h p|=|1,2,\cdots,m,m+1; 0,1,\cdots,p|.$
Substituting them into \eqref{nls-bi-1} one obtains
\begin{align*}
&2(ig_{n,t}f_n-if_{n,t}g_n-g_{n-1}f_{n+1}-g_{n+1}f_{n-1}+2g_{n}f_{n})\\
=&|A|^{-2}(|\h{m},m+2;\h {p-1}||\t {m+1};-1,\h {p-1}|-|\t{m},m+2;-1,\h {p-1}|
|\h {m+1}; \h {p-1}|\\
&-|\t{m+2};\h {p-1}||\h m;-1, \h{p-1}|)\\
&-|A|^{2}(|-1,\h {m};\t {p-1},p+1||\h m;\h p|-
|\h m;\h {p-1},p+1||-1,\h {m};\t {p}|+|\h{m};\t {p+1}||-1,\h{m};\h {p-1}|)\\
=&0,
\end{align*}
in which \eqref{id-1} has been used. Eq.\eqref{nls-bi-2} can  be verified in a similar way.
For \eqref{nls-bi-3} we have
\begin{eqnarray*}
&&f_n^2-g_nh_n-f_{n+1}f_{n-1}\\
&=&|A|^{-2}(|\h m; \h p||\t {m+1};-1,\h {p-1}|+|\h {m+1}; \h {p-1}||\t {m};-1,\h {p}|
-|\t {m+1}; \h{p}||\h {m};-1, \h{p-1}|),
\end{eqnarray*}
which vanishes by means of  \eqref{id-1}.

Suppose that $A=V^{-1}\Gamma V$, i.e. $A$ is similar to $\Gamma$. We introduce $\Phi_n'=V\Phi_n,~ \Psi_n'=V\Psi_n$
which satisfy \eqref{nls-phi-psi} with $\Gamma$ in place of $A$.
Then we have
$f_n(\Phi_n',\Psi_n')=|V|f_n(\Phi_n,\Psi_n)$, $g_n(\Phi_n',\Psi_n')=|V|g_n(\Phi_n,\Psi_n)$ and
$h_n(\Phi_n',\Psi_n')=|V|h_n(\Phi_n,\Psi_n)$,
which means $A$ and its similar form $\Gamma$ generate same $Q_n$ and $R_n$.

\end{appendix}

\end{document}